\documentclass[pra,aps,twocolumn,10pt,showpacs,nofootinbib]{revtex4-2}
\usepackage[utf8]{inputenc}
\usepackage[T1]{fontenc}
\usepackage{amsmath,amsthm,charter,mathptmx}
\usepackage{amsfonts}
\usepackage{amssymb}
\usepackage{mathcmd}
\usepackage{graphicx} 

\usepackage{braket}
\usepackage{bm,bbold}
\usepackage{mathrsfs}
\usepackage{stmaryrd}
\usepackage{esint}
\usepackage{xcolor}

\usepackage{tikz}
\usepackage{rotating} 
\usepackage{float}
\usepackage{qcircuit}

\newtheorem{definition}{Definition}
\newtheorem{theorem}{Theorem}
\newtheorem{lemma}{Lemma}

\usepackage{ulem}
\usepackage[colorlinks=true, linkcolor=green, citecolor=blue]{hyperref}

\begin{document}

\title{Censorship of Quantum Resources in Quantum Networks}
\author{Julien Pinske}
\email{julien.pinske@nbi.ku.dk}
\author{Klaus M\o lmer}
\affiliation{Niels Bohr Institute, University of Copenhagen, Blegdamsvej 17, DK-2100 Copenhagen, Denmark}

\date{\today}

    \begin{abstract}
    We may soon see agencies offering public access to quantum communication networks. In such networks it may be a feature that certain resources are available only to priority users or at a higher user fee, as governed by a protective agency overseeing the communication in the network. If the agency wants to restrict the general users to communicate by transmitting states that we categorize as free states of a quantum resource theory (QRT), it may employ resource-destroying (RD) channels that do not affect the free states. Such channels, however, only exist for the simplest of QRTs, putting fundamental limitations on how quantum resources can be regulated. In this work, we go beyond the present limitation by devising a nonlinear censorship protocol which makes use of classical information about the transmitted state. We study the requirement disabling malicious users from breaking the censorship. 
    The protocol can establish an unbreakable censorship of imaginarity and entanglement, while no such censorship can be made for quantum discord and Bell nonlocality.
    \end{abstract}
    
    \maketitle

    \section{Introduction}
    \label{sec:intro}

    Starting from early academic proposals \cite{I76,W83,S94}, quantum information processing matured into an industrialized field of research, where the emergence of commercial start-ups \cite{C18,B21} and (inter)national programs \cite{RZ23} for building quantum computers, can be witnessed.
    This transformation of modern information societies is driven by novel technological applications including quantum sensing \cite{DM10,D17,A23}, quantum communication \cite{NB22,ZS22,LC23}, quantum money \cite{MO23,HG23}, and quantum machine learning \cite{BW17,BO24}.
    The strive for storing, transmitting, and manipulating quantum information, fuels the prospect of a widely accessible quantum internet \cite{K08,SE18,IC22}.

    Contrariwise, in large public-domain networks establishing certain restrictions on the sharing of quantum resources becomes of increasing interest.
    For once, because quantum information can be used break certain cryptographic schemes \cite{MV18,SC21}.
    This enables malicious parties to carry out cryptographic attacks on critical infrastructures, wherever post-quantum cryptography \cite{JF11,BL17,K22} is not at its state-of-the-art.  
    Moreover, experimental progress is made on using commercial telecommunication lines for transmission of quantum information \cite{ZCB21,LBF22,JX23}.
    Network providers might seek to regulate quantum communication in such networks, e.g., the provider offering free classical communication, but demanding premium fees for the transmission of quantum information.

    To prevent the proliferation of quantum resources, such as coherence and entanglement, governmental agencies and commercial providers might establish a form of \textit{quantum censorship} \cite{PS24}.
    In such a protocol, quantum information which is deemed benign crosses a network unaltered while hazardous quantum information is denied transmission.
    To achieve this, a dominant (governmental or commercial) agency applies a resource-destroying (RD) channel \cite{LH17} locally to each sender-receiver connection (Fig. \ref{fig:motiv}). 
    This ensures that only free states of a quantum resource theory (QRT) \cite{CG19,G24} are transmitted. 
    Unlike resource-breaking \cite{HS03,IK13} and resource-annihilating \cite{MZ10,SB17,BB18} operations, RD channels not only destroy the resource but do not alter a free state. 
    Moreover, RD channels are single-shot operations, thus avoiding costly tomography by the agency.
    A key issue of the censorship is that RD channels do not exist for most resources \cite{G17}.
    For instance, there do not exist RD channels for imaginarity \cite{WK21,XG21,HG18}), quantum entanglement \cite{GG07,HH09,SW13,WD17}, quantum discord \cite{OZ01,MB12}, and non-Gaussianity \cite{ESP02,TZ18,LR18}. 
    At first sight, this appears to limit use-cases of quantum censorship severely.

    In this paper, we overcome this limitation by devising a generalized protocol for quantum censorship that uses, what we term, \textit{conditional RD channels}. 
    These operations are conditioned on additional classical information about the free state to be transmitted by the user.
    This enables the dominant authority to deterministically implement (effectively) nonlinear quantum operations, thus allowing for a wider variety of possible RD maps than previously considered \cite{LH17,G17,PS24}.
    In general, specifying a quantum state by classical means leads to an exponential overhead, thus being computationally unfeasible. 
    However, in the censorship protocol only free states are allowed to be transmitted.
    For the latter often a much more compact and efficient classical description can be given.
    The censorship protocol is termed unbreakable if the resource is destroyed, despite the users being uncooperative, e.g., providing an untruthful description of their state. 
    We give several mathematical statements showing under which circumstances the censorship protocol is unbreakable. 
    In particular, the protocol enables an unbreakable censorship of imaginarity and quantum entanglement, both theories for which no (linear) RD channel exists.
    On the other hand, for quantum discord and Bell nonlocality, we find that censorship can be overcome.

    The structure of the article is as follows. 
    In Sec. \ref{sec:QRT}, we start by reviewing some of the basic concepts in QRT that are relevant to quantum censorship.
    There, we also introduce conditional RD channels, the key notion in our theory.
    Section \ref{sec:GQC} devises a protocol for quantum censorship.
    We provide mathematical statements under which the censorship is secure against uncooperative users in the network.
    A study on the influence of noise on the protocol concludes the section. 
    In Sec. \ref{sec:Ex}, the protocol is applied to several quantum resources, namely, imaginarity, entanglement, discord, and nonlocality.
    Finally, Sec. \ref{sec:FIN} is reserved for a summary of the article and concluding remarks.

    \begin{figure}[t]
    \includegraphics[width=.4\textwidth]{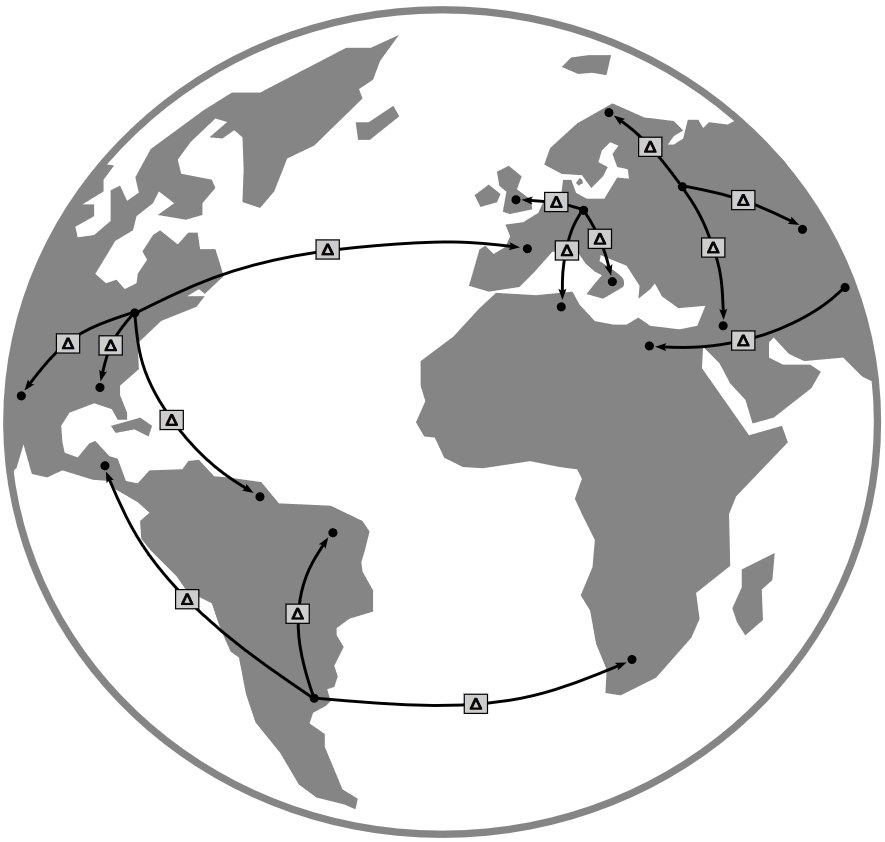}
    \caption{%
        In the quantum censorship protocol, a dominant agency oversees quantum communication in a public-domain network.
        The agency does so by applying a (conditional) RD channel $\Delta$ to a quantum state before sending it to a receiver. 
        The channel destroys any resource present in the state before it reaches the receiver.
    }
    \label{fig:motiv}
    \end{figure}

    \section{Quantum resource theories}
    \label{sec:QRT}

    When trying to censor quantum information, one first has to divide resource states, whose distribution one wants to prevent, from free states, which propagate through the network unchanged.
    Making this distinction is the subject of QRTs \cite{CG19,G24}.
    In a QRT one defines a set of free states $\mathcal{F}(A)$, being a subset of the set of density operators, denoted by $\mathcal{D}(A)$.
    The set $\mathcal{D}(A)$ contains positive(-semidefinite), unit-trace operators $\rho$ acting on the (here, finite-dimensional) Hilbert space $\mathcal{H}_A$ of a system $A$.
    If a state $\rho\notin\mathcal{F}(A)$ is not free, it is said to be a resource.
    There also exist QRTs that consider the resource content of quantum processes \cite{LY20} or quantum invasiveness \cite{MC19}, which we do not study in this work.

    A QRT is said to be convex if for any free state $\sigma^a\in \mathcal{F}(A)$, the convex sum $\sigma=\sum_a t_a\sigma^a$ is again a free state.
    Here, $t_a\geq 0$ and $\sum_a t_a=1$ is a probability distribution.
    The convex hull of an arbitrary set $\mathcal{F}(A)\subseteq\mathcal{D}(A)$ is defined as
    \begin{equation}
        \label{eq:convHull}
        \mathrm{Conv}(\mathcal{F}){=}\Big\{\sum_{a}t_a\sigma^a\,\Big|\,\sigma^a{\in}\mathcal{F}(A),\sum_a t_a{=}1,t_a{\geq}  0\Big\}.
    \end{equation}
    
    Some QRTs satisfy a stronger condition than convexity, that is, they are affine.
    A QRT is said to be affine, if for free states $\sigma^a\in \mathcal{F}(A)$, the state $\sigma=\sum_a t_a\sigma^a$, with $t_a\in\mathbb{R}$ and $\sum_a t_a=1$, is again a free state. 
    For an arbitrary set $\mathcal{F}(A)\subseteq\mathcal{D}(A)$, its affine hull is here defined as 
    \begin{equation}
        \label{eq:affHull}
        \mathrm{Aff}(\mathcal{F}){=}\Big\{\sum_{a}t_a\sigma^a\,\Big|\,\sigma^a{\in}\mathcal{F}(A),\sum_a t_a{=}1\Big\}
        \cap\mathcal D(A).
    \end{equation}
    The intersection with $\mathcal{D}(A)$ ensures that the set $\mathrm{Aff}(\mathcal{F})$ contains states only.
    Note that $\mathrm{Conv}(\mathcal{F})\subseteq \mathrm{Aff}(\mathcal{F})$ holds true, because convex combinations are special cases of affine combinations in which the coefficients $t_a$ are probabilities.

    As an illustrative example, consider the QRT of coherence \cite{BC14,SW18}, in which one quantifies the amount of superpositions with respect to a fixed orthonormal basis $\{\ket{a}\}_a$, the incoherent basis.
    A free (likewise, incoherent) state $\sigma$ admits a diagonal representation in that basis, i.e.,
    $\sigma=\sum_a t_a\ket{a}\bra{a}$. 
    The set of incoherent states $\mathcal{F}(A)$ is affine.
    In particular, for $\sigma$ to be a positive operator, we must have that $t_a\geq 0$.
    Thus, for the case under study the notion of convex and affine coincide.

    \subsection{Tensor-product structures}
    \label{ssec:TPS}

    So far, we only considered a single party $A$.
    As quantum censorship will be employed in large public-domain networks, a formulation of QRTs on compound quantum systems is necessary.
    Let therefore $\mathcal{D}(A_1\dots A_N)$ denote the set of quantum states of an $N$-partite system.
    The convex hull $\mathrm{Conv}\big(\mathcal{D}(A_1)\otimes\dots\otimes\mathcal{D}(A_N)\big)$ corresponds to the set of $N$-partite, fully separable states \cite{W89,H97}.
    More explicitly, an $N$-partite state $\rho$ is said to be separable if it can be written as a convex sum of (pure) product states
    \begin{equation}
        \label{eq:SepState}
        \rho=\sum_b t_b \rho^b_1\otimes\dots\otimes \rho^b_N,
    \end{equation}
    where $t_b\geq 0$ and $\sum_bt_b=1$ is a probability distribution.
    Throughout the paper, we employ the following notation for the tensor products of sets
    \begin{equation}
        \mathcal D(A_1)\otimes\dots\otimes\mathcal D(A_N)=\{\rho_{A_1}\otimes\dots\otimes\rho_{A_n}|\rho_{A_a}\in\mathcal D(A_a)\}.
    \end{equation}

    Given a QRT $\mathcal{F}(A_a)$ on the individual subsystems $a=1,\dots,N$, one has to clarify what constitutes a free state on the composite system.
    The set of $N$-partite free states will be denoted as $\mathcal{F}(A_1\dots A_N)$.
    Most of the physically motivated QRTs admit a tensor-product structure \cite{CG19,LB22}.
    Firstly, this means that the independent preparation of free states $\sigma_{A_1},\dots, \sigma_{A_n}$ gives a free state on the composite system, i.e., $\sigma_{A_1}\otimes\dots \otimes \sigma_{A_N}\in\mathcal{F}(A_1\dots A_N)$.
    Equivalently, we can write
    \begin{equation}
        \mathcal{F}(A_1)\otimes\dots \otimes\mathcal{F}(A_N)\subseteq\mathcal{F}(A_1\dots A_N).
    \end{equation}
    Secondly, discarding subsystems does not create a resource;
    i.e., for $\sigma\in\mathcal{F}(A_1\dots A_N)$, its marginals 
    \begin{equation}
        \mathrm{Tr}_{a}(\sigma)\in \mathcal{F}(A_1\dots A_{a-1}A_{a+1}\dots A_N)
    \end{equation}
    for $a=1,\dots,N$, are free, too.
    
    Therefore, if $\mathcal{F}(A_1),\dots,\mathcal{F}(A_N)$ are convex, then we define
    \begin{equation}
        \label{eq:ConvFree}
        \mathrm{Conv}\big(\mathcal{F}(A_1)\otimes\dots\otimes\mathcal{F}(A_N)\big)\subseteq\mathcal{F}(A_1\dots A_N).
    \end{equation}
    Analogously, if $\mathcal{F}(A_1),\dots,\mathcal{F}(A_N)$ are affine, then we define
    \begin{equation}
        \label{eq:AffFree}
        \mathrm{Aff}\big(\mathcal{F}(A_1)\otimes\dots\otimes\mathcal{F}(A_N)\big)\subseteq \mathcal{F}(A_1\dots A_N).
    \end{equation}
    Unless explicitly stated otherwise, the resources to be considered are assumed to admit a tensor-product structure.

    \subsection{Resource-destroying maps and channels}
    
    Here we formally define RD maps \cite{LH17}, a key notion in quantum censorship.
    A map $\Delta:\mathcal{D}(A)\to \mathcal{D}(B)$ is said to be RD, if it satisfies
    \begin{align}
        \text{(i) }\forall\rho\in\mathcal D(A):& \quad\Delta(\rho)\in\mathcal F(B),
        \tag{resource-destroying}
        \\
        \text{(ii) }\forall\sigma\in\mathcal F(A):& \quad\Delta(\sigma)=\sigma.
        \tag{freeness-preserving}
    \end{align}
    Simply speaking, an RD map outputs a free state, if the input was a resource, and leaves the input unchanged if it was a free state.
    In general, an RD map can be highly nonlinear and does not need to be physically implementable. 
    For condition (ii) to be well-defined, the systems $A$ and $B$ are assumed to be isomorphic, i.e., $\mathcal{H}_A$ and $\mathcal{H}_B$ have the same dimension.
    
    In the standard picture of quantum information processing, physical operations are mathematically expressed as quantum channels \cite{W18}.
    A quantum channel is a linear map $\Lambda:\mathcal{D}(A)\to\mathcal{D}(B)$ that is completely positive, i.e., $\mathrm{id}_C\otimes\Lambda$ maps positive operators onto positive operators.
    Here, $C$ is a reference system of arbitrary size and $\mathrm{id}_C$ denotes the identity.
    If an RD map $\Delta$ is a quantum channel, we refer to it as an RD channel.
    
    For a simple example of an RD channel, consider again the QRT of coherence.
    The (unique) RD channel for this theory is given by a dephasing in the incoherent basis, i.e.,
    \begin{equation}
        \label{eq:RDcoh}
        \Delta(\rho)=\sum_{a}\ket{a}\bra{a}\rho\ket{a}\bra{a}.
    \end{equation}
    The channel $\Delta$ outputs an incoherent state for any $\rho\in\mathcal{D}(A)$.
    It is readily verified that $\Delta$ leaves incoherent states unaltered.
    Here, the set of $N$-partite incoherent states $\mathcal{F}(A_1\dots A_N)=\mathrm{Aff}(\mathcal{F}(A_1)\otimes\dots\otimes \mathcal{F}(A_N))$ is exactly the set of states which is left unaltered by $\Delta^{\otimes N}$ satisfying Eq. \eqref{eq:AffFree}.
    
    A simple physical implementation of the channel \eqref{eq:RDcoh} can be envisaged in terms of polarization optics.
    Here, the incoherent basis is defined by the horizontal and vertical polarization states $\ket{H}$ and $\ket{V}$, respectively. 
    Let $\ket{\psi_A}=\alpha_H\ket{H}+\alpha_V\ket{V}$ be a resource state prepared by a sender $A$.
    The channel $\Delta$ now amounts to applying a polarization filter corresponding to the projective measurement $\ket{H}\bra{H}$ or $\ket{V}\bra{V}$, yielding the state $\ket{H}$ and $\ket{V}$ with probability $|\alpha_H|^2$ and $|\alpha_V|^2$, respectively. 
    As long as it is hidden from a receiver $B$, which filter was used, $B$'s best description of the state is $\sigma_B=|\alpha_H|^2\ket{H}\bra{H}+|\alpha_V|^2\ket{V}\bra{V}$.
    The resource of coherence has been destroyed \cite{PS24}.
 
    While RD channels can be given for the QRTs of coherence \cite{BC14,SW18}, reference frames \cite{BR07,GS08}, and athermality \cite{BH15,SR20}, they do not exist for all resources. 
    For instance, for an RD channel to exist, the set of free states $\mathcal{F}(A)$ must be affine \cite{LH17,G17}.

    To see this, suppose $\mathcal{F}(A)$ is nonaffine.
    Then, there exists an affine combination $\sigma=\sum_a t_a\sigma^a\notin\mathcal{F}(A)$, with $\sigma^a\in\mathcal{F}(A)$ being free. 
    If an RD channel $\Delta$ would exist for a nonaffine theory, then $\Delta\big(\sum_a t_a\sigma^a\big)=\sum_a t_a\Delta(\sigma^a)$ has to be a free state by condition (i).
    On the other hand, condition (ii) implies $\Delta(\sigma^a)=\sigma^a$, thus $\Delta(\sigma)=\sigma$.
    Since, $\sigma$ is not a free state by assumption, we arrive at a contradiction, showing that the channel $\Delta$ does not exist.

    \subsection{Conditional resource-destroying channels}
    
    The above result sets a fundamental limitation on which resources can be destroyed.
    Nevertheless, for a nonaffine QRT, we can introduce a generalization of an RD channel by allowing for a conditioning on the input state. 
    Such a preprocessing (or postprocessing) is a common practice in the field of quantum information to emulate nonlinear dynamics, see, e.g., Refs. \cite{S03,PM17,BB18}.
    In practice, the conditioning will be realized by providing a classical message $m_\sigma$ that specifies (not necessarily uniquely) the free state $\sigma\in\mathcal{F}(A)$. 
    We call such a message $m_\sigma$ a \textit{(state) description} of $\sigma$.
    
    For example, an incoherent state $\sigma$ acting on a $d$-dimensional subspace $\mathcal{H}_A$ is uniquely specified by a probability distribution, i.e., $m_{\sigma}=(p_1,\dots,p_{d-1})$, where $p_d=1-\sum_{a=1}^{d-1}p_a$ does not need to be included in the description. 
    Already, we can observe that free states can often be specified much more compactly, than a general state $\rho\in\mathcal{D}(A)$, for which $d^2-1$ real parameters must be given.
    
    In a quantum-mechanical framework, a description is associated with a pure state $\ket{m_\sigma}$ in an ancilla space $\mathcal{H}_{M}$, which we denote as the \textit{description space} throughout.
    One can imagine the classical information being encoded to (some degree of accuracy) by a string of binaries, e.g., $\ket{m_\sigma}=\ket{010\dots1}$.
    Then, different descriptions can be distinguished by projective measurement, owing to their orthogonality $\braket{m_\sigma|m_{\sigma^\prime}}=\delta_{\sigma \sigma^\prime}$.
    
    Throughout the paper, we will make frequent use of the fact that any $\rho\in\mathcal{D}(MA)$ can be expanded in terms of an orthonormal basis $\ket{m}$ of the description space $\mathcal{H}_M$, viz. 
    \begin{equation}
        \label{eq:exp}
        \rho_{MA}=\sum_{m,n}t_{mn}\ket{m}\bra{n}_M\otimes \rho^{mn}_A.
    \end{equation}
    Since, $\rho_{MA}$ is a positive unit-trace operator, it is necessary that the coefficients satisfy $t_{mm}\geq 0$, $\sum_{m}t_{mm}=1$, and $\rho^{mm}\in\mathcal{D}(A)$ must be quantum states.
    Note that a general state $\rho_{MA}$ in Eq. \eqref{eq:exp} can also include superpositions of message states $\ket{m}$.  
    
    With this notation established we introduce the notion of a \textit{conditional RD channel}.
    \begin{definition}
        \label{def:CRD}
        A quantum channel $\Delta:\mathcal{D}(MA)\to\mathcal{D}(B)$ is said to be a conditional RD channel, if it satisfies
    \begin{align*}
           \mathrm{(iii)}\quad\forall\rho\in\mathcal D(MA):&\quad\Delta(\rho)\in\mathcal F(B),\\
            \mathrm{(iv)}\quad\forall\sigma\in\mathcal F(A):&\quad\Delta(\ket{m_\sigma}\bra{m_\sigma}\otimes\sigma)=\sigma.
    \end{align*} 
    \end{definition}
    The channel $\Delta$ destroys any resource, independent of the input, and preserves a free state $\sigma$, when given its description $m_\sigma$.
    To censor quantum information using conditional RD channels is the key purpose of this paper.  

    Conditional RD channels have a useful representation in terms of a projective measurement on $M$, viz.,
    \begin{equation}
        \label{eq:CRD-Rep}
        \Delta=\sum_{m_\sigma}\braket{m_\sigma|(\,\cdot\,)|m_\sigma}\otimes \Delta_{m_{\sigma}}.
    \end{equation}
    In order for $\Delta$ in Eq. \eqref{eq:CRD-Rep} to be a conditional RD channel, we must have that
    \begin{align*}
        \text{(v)}\quad\forall\rho\in\mathcal D(A):& \quad\Delta_{m_\sigma}(\rho)\in\mathcal{F}(B),
        \\
        \text{(vi)}\quad\forall\sigma\in\mathcal F(A):& \quad\Delta_{m_{\sigma}}(\sigma)=\sigma.
    \end{align*}
    In Appendix \ref{app:rep}, we show that Eq. \eqref{eq:CRD-Rep} is not only a sufficient but also a necessary condition for $\Delta$ to be a conditional RD channel.
    Simply speaking, the channel $\Delta_{m_\sigma}$ destroy any resource state, independent of the given description, but preserves a free state $\sigma$, when given the (correct) description $m_\sigma$.
    As an information processing protocol $\Delta$ realizes the circuit
    \begin{equation*}
        \Qcircuit @C=1em @R=.7em {
            \lstick{m}& \qw & \meter & \cctrlo{1} & \rstick{}\\
            \lstick{A}& \qw & \qw & \gate{\Delta_m} & \qw & \rstick{B.}\qw\\ 
        }
    \end{equation*}
    Note that a conditional RD channel does only exist for convex QRTs, as given in the following theorem.
    \begin{theorem}
        \label{th:CRD}
        Let $\mathcal{F}$ be nonconvex.
        Then, there does not exist a conditional RD channel $\Delta$.
    \end{theorem}    
    \begin{proof}
        Consider the state
        \begin{equation}
            \rho_{MA}=\sum_{a}p_a\ket{m_{\sigma^a}}\bra{m_{\sigma^a}}\otimes \sigma^a,
        \end{equation}
        with $\sigma^a\in\mathcal{F}(A)$ being free states and the convex sum $\sum_a p_a\sigma^a\notin\mathcal{F}(A)$ being a resource.
        Such a probabilistic mixture must exist, because $\mathcal{F}(A)$ is nonconvex.
        Now, suppose there exists a conditional RD channel $\Delta$ for this QRT.
        From the representation in Eq. \eqref{eq:CRD-Rep} it follows that
        \begin{equation}
            \begin{split}
                \Delta(\rho)&=\sum_{a,b} p_a|\braket{m_{\sigma^a}|m_{\sigma^b}}|^2\Delta_{m_{\sigma^b}}(\sigma^a),\\
                &=\sum_a p_a\sigma^a,
            \end{split}
        \end{equation}
        where we made use of $\braket{m_{\sigma^a}|m_{\sigma^b}}=\delta_{ab}$ and $\Delta_{m_{\sigma^a}}(\sigma^a)=\sigma^a$ by condition (vi).
        Then, $\Delta(\rho)=\sum_a p_a\sigma^a\notin\mathcal{F}(B)$ is a resource state by initial assumption.
        This shows that $\Delta$ fails to satisfy (iii), thus leaving us with a contradiction.
    \end{proof}

    Intuitively, this failure of destroying nonconvex resources follows from the projective measurement on the description space $M$ introducing additional randomness into the system $A$.
    Since randomness constitutes a resource in nonconvex theories, this resource is contained in the output state.
    Note that conditional RD channels give us a more general way of destroying resources than regular RD channels.
    The former apply not only to affine QRTs, but also for convex QRTs, such as entanglement.

    \subsubsection{Examples of conditional RD channels}
    
    Any RD channel is a conditional RD channel for which the description space $\mathcal{H}_M$ is one-dimensional.
    In other words, the channel acts independent of any given description.
    For example, the dephasing channel \eqref{eq:RDcoh} preserves incoherent states without a user specifying which incoherent state was send.

    The simplest example of a conditional RD channel
    is given by choosing $\Delta_{m_\sigma}$ in Eq. \eqref{eq:CRD-Rep} as a replacement channel
    \begin{equation}
        \label{eq:repl}
        \Delta_{m_\sigma}(\rho)=\mathrm{Tr}(\rho)\sigma.
    \end{equation}
    It is not hard to see, that \eqref{eq:repl} fulfills the relations (v) and (vi). 
    For the replacement channel \eqref{eq:repl}, a $d\times d$ density operator $\sigma$ must be specified uniquely in the description $m_\sigma$.
    In general, this requires $d^2-1$ real numbers to be contained in the description $m_\sigma$.
    In many specific situations, the number of parameters might actually be much lower, because free states are usually characterized by fewer parameters than a general quantum state.
    Generally speaking, it is desirable to search for compact state descriptions.
    This will have the effect that the orthogonality relation is modified to $\braket{m_\sigma|m_{\sigma^\prime}}=\delta_{[\sigma] [\sigma^\prime]}$, where $[\sigma]$ is some equivalence class of free states all having the same description $m_\sigma$.

    A more intriguing example of a conditional RD channel is a dephasing with respect to a rotated basis, which is specified in the description. 
    More specifically, let
    \begin{equation}
        \label{eq:deph}
        \Delta_{m_\sigma}(\rho)=\sum_{a}\ket{\sigma^a}\bra{\sigma^a}\rho\ket{\sigma^a}\bra{\sigma^a},
    \end{equation}
    where $\{\ket{\sigma^a}\}_a$ are the eigenvectors of the state $\sigma$ given in the description $m_\sigma$. 
    The channel $\Delta_{m_\sigma}$ satisfies (v) and (vi) only if the eigenstates $\ket{\sigma^a}$ are free as well, i.e., $\ket{\sigma^a}\bra{\sigma^a}\in\mathcal{F}(A)$.
    Two prominent resources, for whom this is the case, are coherence \cite{BC14} and imaginarity \cite{WK21} (see Sec. \ref{sec:Ex}), while it is not the case for entanglement \cite{HH09}.
    There, a free (i.e., separable) state can have eigenstates which are resourceful (i.e., entangled).

    Let $\rho=\sum_{b}\lambda_b\ket{\phi^b}\bra{\phi^b}$ be the spectral resolution of an arbitrary state and $m_\sigma$ be the description of a free state $\sigma\in\mathcal{F}(A)$. 
    Then we have
    \begin{equation}
        \label{eq:CheckRD}
        \begin{split}
            \Delta_{m_\sigma}(\rho)&=\sum_{a}\ket{\sigma^a}\bra{\sigma^a}\rho\ket{\sigma^a}\bra{\sigma^a},\\
            &=\sum_{a,b}\lambda_b \vert\braket{\sigma^a|\phi^b}\vert^2\ket{\sigma^a}\bra{\sigma^a},\\
            &=\sum_{a}\nu_a\ket{\sigma^a}\bra{\sigma^a},
        \end{split}
    \end{equation}
    where we introduced probabilities $\nu_a=\sum_b \lambda_b \vert\braket{\sigma^a|\phi^b}\vert^2$.
    The last line of Eq. \eqref{eq:CheckRD} is a free state, because it is to a convex sum of the free states $\ket{\sigma^a}\bra{\sigma^a}$.
    Hence, $\Delta_{m_\sigma}(\rho)\in\mathcal{F}(B)$ is free and condition (v) is satisfied.
    In particular, let $\rho=\sigma$ be the free state whose description $m_\sigma$ was given.
    Then, $\ket{\phi^b}=\ket{\sigma^b}$ and we have
    $\braket{\sigma^a|\phi^b}=\braket{\sigma^a|\sigma^b}=\delta_{ab}$ in Eq. \eqref{eq:CheckRD}.
    This yields
    \begin{equation}
        \begin{split}
            \Delta_{m_\sigma}(\sigma)&=\sum_{a,b}\lambda_b \delta_{ab}\ket{\sigma^a}\bra{\sigma^a},\\
            &=\sigma.\\
        \end{split}
    \end{equation}
    Thus, condition (vi) is fulfilled as well.
    
    For the replacement channel \eqref{eq:repl} the description was unique, i.e., $[\sigma]=\sigma$. 
    For the the channel \eqref{eq:deph}, we have an equivalence class of states with the same eigenvectors, because in Eq. \eqref{eq:deph} no information about the eigenvalues of $\sigma$ is required.
    The equivalence class reads
    \begin{equation}
        [\sigma]=\big\{\sigma^\prime\in\mathcal{F}(A)\,\big|\,[\sigma,\sigma^\prime]=0\big\}.
    \end{equation}
    Here, we used that the commutator $[\sigma,\sigma^\prime]=\sigma\sigma^\prime-\sigma^\prime\sigma$ vanishes if and only if $\sigma$ and $\sigma^\prime$ have the same eigenvectors.

    \section{Conditional quantum censorship}
    \label{sec:GQC}
    In the following, a protocol for quantum censorship is devised. 
    The protocol utilizes conditional RD channels to establish censorship in a network. 
    Firstly, the protocol is studied for noiseless channels. 
    In particular, we discuss under which circumstances the censorship cannot be overcome by uncooperative users in the network.  
    Finally, the effect of noise on the protocol is investigated. 

    \subsection{Censorship over noiseless channels}
    
    Consider $N$ senders $A_1,\dots,A_N$ who have access to local quantum resources, e.g., party $A_a$ can prepare any state $\rho_{A_a}\in \mathcal{D}(A_a)$. 
    In an unregulated network, each sender is connected to one of the receivers $B_1,\dots,B_N$ via the noiseless channel $\mathrm{id}_{A_a\to B_a}$.
    However, in order to prevent the transmission of resource states, an agent sits in between each sender-receiver pair.
    The agent's goal is to limit the type of quantum states that can be shared between parties to the free states $\mathcal{F}(A_a)$ of a QRT.
    The agent informs each sender that upon transmission of their state $\sigma$ a truthful state description $m_\sigma$ is to be send as well.
    This establishes the user agreement of the public-domain network.
    To enforce that policy, the agent applies a conditional RD channel $\Delta$.
    Thus, the general information processing protocol of (noiseless) quantum censorship is 
    \begin{equation*}
        \Qcircuit @C=1em @R=.7em {
            \lstick{}& \qw & \meter & \cctrlo{1} & \rstick{}\\
            \lstick{}& \qw & \qw & \gate{\Delta_m} & \qw & \rstick{B_1}\qw\inputgroup{1}{2}{1.1em}{A_1}\\ 
            \lstick{\vdots} & &\vdots & & \vdots\\
    		\lstick{} & & & & \\
            \lstick{}& \qw & \meter & \cctrlo{1} & \rstick{}\\
            \lstick{}& \qw & \qw & \gate{\Delta_m} & \qw & \rstick{B_N.}\qw\inputgroup{5}{6}{1.1em}{A_N}\\ 
        }
    \end{equation*}
    
    Suppose, each sender tries to transmit a resource state, i.e., the composite system is in the product state 
    \begin{equation}
        \rho^1\otimes\dots\otimes\rho^N\notin\mathcal{F}(A_1)\otimes\dots \otimes\mathcal{F}(A_N).
    \end{equation}
    They then give some (untruthful) state description $m_{\sigma^1},\dots,m_{\sigma^N}$ to be transmitted together with their states.
    After censorship, the receiving parties obtain [see (v)]
    \begin{equation}
        \Delta_{m_{\sigma^1}}(\rho^1)\otimes\dots\otimes \Delta_{m_{\sigma^N}}(\rho^N)\in\mathcal{F}(B_1)\otimes\dots \otimes\mathcal{F}(B_N).
    \end{equation}
    The result is a free state, independent of whether the senders gave truthful state descriptions or not.
    As intended by the agent, any potential resource in the network has been destroyed before reaching one of the receivers $B_1,\dots,B_N$.
    On the other hand, if the initial state $\sigma^1\otimes\dots\otimes\sigma^N$ is free, it remains unchanged by the censorship [see (vi)]
    \begin{equation}
        \Delta_{m_{\sigma^1}}(\sigma^1)\otimes \dots\otimes \Delta_{m_{\sigma^N}}(\sigma^N)=\sigma^1\otimes\dots\otimes\sigma^N,
    \end{equation}
    when given the correct descriptions $m_{\sigma^1},\dots,m_{\sigma^N}$.
    This allows the users of the network to carry out (uninterrupted) communication only with free states.

    Most protocols used for quantum-information processing that involve nonlinear channels are not scalable, because they demand for exponential classical resources.
    This can be due to a heralding process having low success probability \cite{KLM01} or by assuming classical information about a quantum state \cite{HC23}, which scales exponentially with the number of qudits (likewise the dimension of the system).
    Interestingly, this does not pose a fundamental challenge to 
    the censorship protocol, because we expect that there often exist compact descriptions of the set of free states for most resources.
    In particular, resourceful states do not need to be specified, as they are not supposed to pass the network in the first place.

    \subsection{Breakable and unbreakable censorship}

    A central question in quantum censorship is, if it is possible to smuggle a resource past the agent \cite{PS24}. 
    This could be done for instance, by a sender injecting part of their quantum state into the description space or by multiple senders sharing resources before transmission.
    Before delving into these questions, we define when censorship has been broken \cite{PS24}.
    \begin{definition}
        \label{def:break}
        A censorship is breakable if there exists a state $\rho\in\mathcal{D}(M_1A_1\dots M_NA_N)$ such that $\Delta^{\otimes N}(\rho)\notin\mathcal{F}(B_1\dots B_N)$.
        Otherwise, censorship is said to be unbreakable.
    \end{definition}
    Simply speaking, censorship is breakable if a resource state reaches the receivers.
    Once this occurred, the receivers could coordinate their actions to make use of the resource.

    By virtue of Definition \ref{def:CRD}, it is not possible for a single sender $A$ to overcome the censorship, because any state $\rho_{MA}\in\mathcal{D}(M A)$ is mapped onto a free state $\Delta(\rho)\in\mathcal{F}(B)$ [see (iii)]. 
    Next, suppose a group of senders $A_1,\dots,A_N$ coordinate their actions to prepare an $N$-partite state $\rho\in\mathcal{D}(M_1A_1\dots M_NA_N)$.
    The agency wants to ensure that no resource state passes through the network.
    Formally, this corresponds to the question whether the channel $\Delta^{\otimes N}$ will always output a free state on the composite system?
    The answer will drastically vary depending on what the underlying structure of the QRT is.

    \subsubsection{Censorship of affine resource theories}
    
    For an affine QRT, the censorship turns out to be unbreakable as is established by the following theorem.
    \begin{theorem}
        \label{th:CenAff}
        If $\Delta$ is a conditional RD channel of an affine QRT $\mathcal{F}$,
        then censorship is unbreakable.
    \end{theorem}
    \begin{proof}
        Let $\rho\in\mathcal{D}(M_1 A_1\dots M_N A_N)$ be an arbitrary state of $N$ senders. 
        The state can be expanded as [see Eq. \eqref{eq:exp}]
        \begin{equation}
            \label{eq:FullExp}
            \begin{split}
                \rho &=\sum_{a_1\dots b_N}t_{a_1\dots b_N}\ket{m_{\sigma^{a_1}}\dots m_{\sigma^{a_N}}}\bra{m_{\sigma^{b_1}}\dots m_{\sigma^{b_N}}}\otimes \rho^{a_1 \dots b_N},
            \end{split}
        \end{equation}
        with coefficients satisfying $\sum_{a_1\dots b_N}t_{a_1\dots b_N}=1$ and $\rho^{a_1\dots b_N}$ being $N$-partite operators.
        After applying the censorship, the receiving parties are left with the state
        \begin{equation}
            \label{eq:proof1}
            \begin{split}
                \Delta^{\otimes N}(\rho)
                &=\sum_{a_1\dots a_N}t_{a_1\dots a_N} \big(\Delta_{m_{\sigma^{a_1}}}\otimes \dots \otimes\Delta_{m_{\sigma^{a_N}}}\big)(\rho^{a_1\dots a_N}),\\
                \end{split}
        \end{equation}
        where we made use of the orthogonality relation $\braket{m_{\sigma^a}|m_{\sigma^b}}=\delta_{ab}$.  
        Because the QRT $\mathcal{F}(A_a)$ is affine, its composite $\mathcal{F}(A_1\dots A_N)$ contains the affine hull \eqref{eq:AffFree}.
        Thus, $\Delta^{\otimes N}(\rho)$ is a free state, whenever $\big(\Delta_{m_{\sigma^{a_1}}}\otimes \dots \otimes\Delta_{m_{\sigma^{a_N}}}\big)(\rho^{a_1\dots a_N})$ is. 
        This is the case, as can be shown by expanding the state as an affine combination of product states, viz. 
        \begin{equation}
            \label{eq:basis}
            \rho^{a_1\dots a_N}=\sum_{b_1\dots b_N}s_{b_1\dots b_N}^{a_1\dots a_N}\omega^{b_1}\otimes\dots \otimes\omega^{b_N},
        \end{equation}
        where the states $\omega^{b_1}\otimes\dots \otimes\omega^{b_N}$ forms a basis of $\mathcal{D}(A_1\dots A_N)$. 
        It follows from condition (v), that
        \begin{equation}
            \Delta_{m_{\sigma^{a_1}}}\big(\omega^{b_1}\big),\dots, \Delta_{m_{\sigma^{a_N}}}\big(\omega^{b_N}\big),
        \end{equation}
        are free states. 
        Since, $\mathcal{F}(A_1\dots A_N)$ is affine, 
        \begin{equation}
            \big(\Delta_{m_{\sigma^{a_1}}}\otimes \dots \otimes\Delta_{m_{\sigma^{a_N}}}\big)(\rho^{a_1\dots a_N})\in\mathcal{F}(B_1\dots B_N)
        \end{equation}
        must be free as well, because Eq. \eqref{eq:basis} is an affine combination. 
        In conclusion, $\Delta^{\otimes N}(\rho)$ is a free state, thus proving the assertion.
    \end{proof}
    
    The above theorem is more general than previous results \cite{PS24}, as it holds even for affine QRTs for which no (linear) RD channel exists.
    For example, there does not exist a RD channel for the QRT of the imaginary \cite{G17}. 
    However, as will be shown in Sec. \ref{sec:Ex}, an unbreakable censorship for this resource can be established using a conditional RD channel.

    \subsubsection{Censorship of convex resource theories}

    When confronted with the censorship of a convex QRT, the set $\mathcal{F}(A_1\dots A_N)$, by definition, includes the convex hull \eqref{eq:ConvFree}.
    The set contains free, $N$-partite separable states. 
    Therefore, entanglement-breaking channels play an important role in their censorship. 
    A channel $\Lambda:\mathcal{D}(A)\to\mathcal{D}(B)$ is entanglement breaking \cite{HS03,W18}, if $\big(\mathrm{id}_C\otimes\Lambda\big)(\rho)$ is a separable state for any $\rho\in\mathcal{D}(CA)$.
    Here, $C$ is a reference system of arbitrary size.
    In more technical language, $\mathrm{id}_C\otimes\Lambda$ maps elements in $\mathcal{D}(CA)$ onto elements in the convex hull $\mathrm{Conv}\big(\mathcal{D}(C)\otimes\mathcal{D}(A)\big)$.
    Note that a quantum channel $\Lambda$ is entanglement breaking if and only if it has an operator-sum representation $\Lambda=\sum_a K_a(\,\cdot\,)K_a^\dagger$, with rank-one Kraus operators $K_a$ \cite{HS03}. 
    For example, the dephasing channels in Eqs. \eqref{eq:RDcoh} and \eqref{eq:deph} are entanglement breaking.
    Using the above terminology we have the following theorem.
    \begin{theorem}
    \label{th:CenCon}
        If $\Delta$ is an entanglement-breaking conditional RD channel of a convex QRT $\mathcal{F}$, then censorship is unbreakable.
    \end{theorem}
    \begin{proof}
        Recall that censorship is breakable if there is a state $\rho\in\mathcal{D}(M_1A_1\dots M_NA_N)$ such that a resource state $\Delta^{\otimes N}(\rho)\notin \mathcal{F}(B_1\dots B_N)$ reaches the receivers; see Definition \ref{def:CRD}. 
        Using a similar expansion as in Eq. \eqref{eq:FullExp} shows that
        \begin{equation}
            \label{eq:proof2}
            \begin{split}
                \Delta^{\otimes N}(\rho)
                &=\sum_{a_1\dots a_N}t_{a_1\dots a_N} \big(\Delta_{m_{\sigma^{a_1}}}\otimes \dots \otimes\Delta_{m_{\sigma^{a_N}}}\big)(\rho^{a_1\dots a_N}).\\
                \end{split}
        \end{equation}
        We now show that the sum in Eq. \eqref{eq:proof2} is convex, i.e., $t_{a_1\dots a_N}\geq0$. 
        This can be seen by computing the reduced quantum state 
        \begin{equation}
            \mathrm{Tr}_{M_1\dots M_N}(\rho)=\sum_{a_1\dots a_N}t_{a_1\dots a_N} \rho^{a_1\dots a_N},
        \end{equation}
        where $\rho^{a_1\dots a_N}\in\mathcal{D}(A_1\dots A_N)$ can be arbitrary quantum states.
        Hence, in order for $\mathrm{Tr}_{M_1\dots M_N}(\rho)$ to be a quantum state, we must have that $t_{a_1\dots a_N}\geq0$ are probabilities [see Eq. \eqref{eq:exp} as well].
        Thus, the sum in Eq. \eqref{eq:proof2} is convex.
        Because the QRT $\mathcal{F}(A_a)$ is convex, its composite $\mathcal{F}(A_1\dots A_N)$ contains the convex hull \eqref{eq:ConvFree}.
        It remains to show that $\big(\Delta_{m_{\sigma^{a_1}}}\otimes \dots \otimes\Delta_{m_{\sigma^{a_N}}}\big)(\rho^{a_1\dots a_N})$ is a free state. 
        By assumption $\Delta$ is entanglement breaking. 
        Thus, it maps onto separable states and can be written as a convex sum
        \begin{equation*}
            \begin{split}
                \big(\Delta_{m_{\sigma^{a_1}}}\otimes \dots \otimes\Delta_{m_{\sigma^{a_N}}}\big)(\rho^{a_1\dots a_N})&=\sum_{b_1\dots b_N}q_{b_1\dots b_N}^{a_1\dots a_N}\omega^{b_1}\otimes\dots \otimes\omega^{b_N},\\
                &\in\mathrm{Conv}\big(\mathcal{D}(A_1)\otimes\dots\otimes\mathcal{D}(A_N)\big),\\
            \end{split}
        \end{equation*}
        with probabilities $q_{b_1\dots b_N}^{a_1\dots a_N}\geq 0$. 
        Applying $\Delta_{m_\sigma}$ again, we get
        \begin{equation}
            \begin{split}
                &\big(\Delta_{m_{\sigma^{a_1}}}\otimes \dots \otimes\Delta_{m_{\sigma^{a_N}}}\big)^2(\rho^{a_1\dots a_N})\\
                &=\big(\Delta_{m_{\sigma^{a_1}}}\otimes \dots \otimes\Delta_{m_{\sigma^{a_N}}}\big)(\rho^{a_1\dots a_N}),
            \end{split}
        \end{equation}
        because $\Delta_{m_\sigma}^2=\Delta_{m_\sigma}$ [see (v) and (vi)]. 
        This implies $\Delta_{m_\sigma}(\omega^{b_c})=\omega^{b_c}$ for all $c=1,\dots,N$. 
        Hence, $\omega^{b_c}\in\mathcal{F}(B_c)$ is a free state.
        It is then clear that $\big(\Delta_{m_{\sigma^{a_1}}}\otimes \dots \otimes\Delta_{m_{\sigma^{a_N}}}\big)(\rho^{a_1\dots a_N})$ is a convex sum of free states. 
        Hence, $\big(\Delta_{m_{\sigma^{a_1}}}\otimes \dots \otimes\Delta_{m_{\sigma^{a_N}}}\big)(\rho^{a_1\dots a_N})\in\mathcal{F}(B_1\dots B_N)$ is free due to Eq. \eqref{eq:ConvFree}.
        In conclusion, $\Delta^{\otimes N}(\rho)$ is a free state, thus proving the assertion.
    \end{proof}

    The above result is crucial in setting up a censorship of entanglement, as will be done in Sec. \ref{sec:Ex}.
    In this QRT, the set of free (i.e., separable) states is convex, but not affine.
    Therefore, Theorem \ref{th:CenCon} allows us to establish an unbreakable censorship.
    In summary, the censorship discussed here, remains unbreakable for a much wider variety of situations than what can be accomplished with (unconditional) RD channels.

    \subsubsection{Censorship of resources that can be activated}

    So far, we only considered resources that admit a tensor-product structure (see Sec. \ref{ssec:TPS}).
    Nevertheless, there are more exotic types of resources, such as Bell nonlocality \cite{B64}, that do not admit a tensor-product structure.
    In such theories, resources can be activated \cite{P12}, i.e., the parallel preparation of free states $\sigma^1,\dots,\sigma^N$ of the individual senders can result in a resource state on the composite system.
    In formula
    \begin{equation}
        \sigma^1\otimes\dots\otimes\sigma^N\notin\mathcal{F}(A_1\dots A_N).
    \end{equation}
    Equivalently, we can write 
    \begin{equation}
        \mathcal{F}(A_1)\otimes\dots\otimes\mathcal{F}(A_N)\nsubseteq\mathcal{F}(A_1\dots A_N).
    \end{equation}
    For an example of resource activation see Sec. \ref{ssec:Bell} or Ref. \cite{P12}.
    These resources pose a particular challenge to a decentralized censorship protocol, because it cannot be verified locally, whether the network contains a resource.
    More precisely, the censorship is breakable.

    \begin{theorem}
        \label{th:TPS}
        Let $\sigma^1\in\mathcal{F}(A_1),\dots, \sigma^N\in\mathcal{F}(A_N)$ be free states of a QRT that can be activated.
        Then censorship is breakable.
    \end{theorem}
    \begin{proof}
        Recall that censorship is breakable if a  there is a resource state $\rho\in\mathcal{D}(M_1 A_1\dots M_N A_N)$ that reaches the receivers, i.e., $\Delta^{\otimes N}(\rho)\notin \mathcal{F}(B_1\dots B_N)$ (Definition \ref{def:break}).
        Here, choose
        \begin{equation}             \rho=\ket{m_{\sigma^1}}\bra{m_{\sigma^1}}\otimes\sigma^1\dots \otimes\ket{m_{\sigma^N}}\bra{m_{\sigma^N}}\otimes\sigma^N.
        \end{equation}
        The censorship protocol acts locally onto each sender's state resulting in
        \begin{equation}
                \Delta^{\otimes N}\big(\rho\big)=\sigma^1\otimes\dots\otimes\sigma^N,
        \end{equation}
        where property (iv) was used.
        Since, the state $\sigma^1\otimes\dots\otimes\sigma^N\notin\mathcal{F}(A_1\dots A_N)$ is a resource, censorship has been broken.
    \end{proof}
    Resources that can be activated can thus not be regulated in a public-domain quantum network.

    \subsection{Censorship over noisy channels}

    The above discussion was concerned only with noiseless communication;
    that is, each sender is connected to a receiver via an identity channel.
    In any realistic situation, we expect communication to be performed over a noisy channel $\Phi$.
    In the context of the censorship protocol, one might be worried that the conditional RD channel $\Delta$ introduces additional changes.
    Thus, we consider the noisy process $\Phi$ to occur before the channel $\Delta$.
    The information-processing protocol for the noisy case reads
    \begin{equation*}
        \Qcircuit @C=1em @R=.7em {
            \lstick{}& \qw & \meter & \cctrlo{1} & \rstick{}\\
            \lstick{}& \gate{\Phi} & \qw & \gate{\Delta_m} & \qw & \rstick{B_1}\qw\inputgroup{1}{2}{1.1em}{A_1}\\ 
            \lstick{\vdots} & &\vdots & & \vdots\\
    		\lstick{} & & & & \\
            \lstick{}& \qw & \meter & \cctrlo{1} & \rstick{}\\
            \lstick{}& \gate{\Phi} & \qw & \gate{\Delta_m} & \qw & \rstick{B_N.}\qw\inputgroup{5}{6}{1.1em}{A_N}\\ 
        }
    \end{equation*}

	Throughout, we make the reasonable assumption that the noise does not create a resource from a free state, i.e., $\Phi$ is resource non-generating \cite{CG19,LH17}.
    Formally, for any free state $\sigma\in\mathcal{F}(A_a)$, one has $\Phi(\sigma)\in\mathcal{F}(A_a)$.
    From a resource-theoretical viewpoint, resource non-generating channels form the largest class of channels that can be considered to be the free channels of a QRT \cite{CFS16}.
    Moreover, we assume the state description $m_\sigma$ to be transmitted via a noiseless classical channel, and to be read out correctly.
    This seems reasonable as there exist well established and scalable error-correcting and fault-tolerance procedures for classical communication \cite{PH98}.
    
    If censorship is established via an (unconditional) RD channel $\Delta$ (i.e., the message space $M$ is one-dimensional), then $\Delta(\sigma)=\sigma$ for any $\sigma\in\mathcal{F}(A)$.
	This implies that the noise $\Phi$ commutes with $\Delta$ on the set of free states, i.e., 
    \begin{equation}
        \forall \sigma\in\mathcal{F}(A): \quad(\Delta\circ\Phi)(\sigma)=(\Phi\circ\Delta)(\sigma)=\Phi(\sigma).
    \end{equation}
	The noisy state $\Phi(\sigma)$ reaches the receiver, and the censorship did not introduce undesirable changes through $\Delta$ \cite{PS24}.

    The situation is more involved if we consider censorship using a conditional RD channel \eqref{eq:CRD-Rep}. 
    This is due to the reason that after the noise occurs, the state description $m_\sigma$ differs from the state $\Phi(\sigma)$ being censored.
    Thus,
    \begin{equation}
        \Delta(\ket{m_\sigma}\bra{m_\sigma}\otimes \Phi(\sigma))\neq \Phi(\sigma)
    \end{equation}
    is generally not preserved. 
    It would of course be desirable if the output state $\Delta_{m_\sigma}(\Phi(\sigma))$ is closer to the input state $\sigma$ than the noisy state $\Phi(\sigma)$. 
    Thus, removing part of the noise. 
    
    As a case study, consider the channel $\Delta_{m_\sigma}$ in Eq. \eqref{eq:deph}.
    It implements a dephasing with respect to the eigenvectors of $\sigma$.
    Let $\Vert A\Vert=\sqrt{\mathrm{Tr}(A^\dag A)}$ denote the Hilbert-Schmidt norm. 
    Furthermore, let $\sigma=\sum_a\lambda_a\ket{\sigma^a}\bra{\sigma^a}$ and $\Phi(\sigma)=\sum_a\mu_a\ket{\phi^a}\bra{\phi^a}$ be spectral resolutions. 
    Their distance is computed to
    \begin{equation}
       \begin{split}
           \Vert\sigma-\Phi(\sigma)\Vert^2&=\sum_a\lambda_a^2-2\mathrm{Tr}\big(\sigma\Phi(\sigma)\big)+\sum_a\mu_a^2,\\
           &=\sum_a\lambda_a^2-2\sum_{a}\lambda_a\nu_a+\sum_a\mu_a^2,
       \end{split}
    \end{equation}
    where we defined probabilities $\nu_a=\sum_{b}\mu_b|\braket{\sigma^b|\phi^a}|^2$.
    Next, the state after censorship is found to be
    \begin{equation}
        \Delta_{m_\sigma}(\Phi(\sigma))=\sum_{a,b}\mu_a|\braket{\sigma^b|\phi^a}|^2\ket{\sigma^b}\bra{\sigma^b}.
    \end{equation}
    We then have
    \begin{equation}
       \begin{split}
           \Vert\sigma-\Delta_{m_\sigma}(\Phi(\sigma))\Vert^2&=\sum_a\lambda_a^2-2\mathrm{Tr}\big[\sigma\Delta_{m_\sigma}\big(\Phi(\sigma)\big)\big]+\sum_a\nu_a^2,\\
           &=\sum_a\lambda_a^2-2\sum_{a}\lambda_a\nu_a+\sum_a\nu_a^2.
       \end{split}
    \end{equation}
    Noting that the purity $\mathrm{Tr}\big[\Phi(\sigma)^2\big]=\sum_a \mu_a^2$ of any state $\Phi(\sigma)$ is non-increasing under the (unital) dephasing channel $\Delta_{m_\sigma}$ (see Appendix \ref{app:pur} for a proof), it follows that 
    \begin{equation}
    \label{eq:non-inc}
    \sum_a\nu_a^2\leq\sum_a\mu_a^2.
    \end{equation}
    Comparing the computed distances yields
    \begin{equation}
        \Vert\sigma-\Delta_{m_\sigma}(\Phi(\sigma))\Vert \leq \Vert\sigma-\Phi(\sigma)\Vert.
    \end{equation}
    The conditional RD channel \eqref{eq:deph} leads to a correction towards the free state $\sigma$ that was originally to be transmitted. 
    It should be stressed that the above argument relies on the assumption that the state description $m_\sigma$ remained noiseless during transmission. 
    The result does not extend to the case in which the state description is subject to noise as well.

    In summary, when an agent implements the censorship protocol using a conditional RD channel, a robust implementation of the state description, possibly via classical electronics or photonics, not only ensures that the censorship remains functional, but might have a noise-suppressing effect as well. 
    Thus, improving the user experience in the public-domain network.

    \section{Censorship of specific resources}
    \label{sec:Ex}
    
     In the following, the censorship protocol is applied to several resources for which no RD channel exists.
     This includes imaginarity, entanglement, discord, and nonlocality.
     We show that the censorship protocol of imaginarity and entanglement is unbreakable.
     In contrast, we show that discord and nonlocality can be smuggled past the agent in a systematic way.

    \subsection{Quantum censorship of the imaginary}
    
    Quantum theory is formulated using complex numbers.
    In the QRT of imaginarity \cite{WK21,XG21,HG18} the set of free states consists of all density operators whose imaginary part vanishes, i.e.,
    \begin{equation}
        \mathcal{F}(A)=\big\{\sigma\in\mathcal{D}(A)\,\big|\,\braket{a|\sigma|b}\in\mathbb{R}\big\}.
    \end{equation}
    Similar to coherence, imaginarity is defined with respect to a reference basis $\{\ket{a}\}_a$.
    To put forward the imaginary as a resource is motivated by several applications in discrimination tasks \cite{Z21}, pseudo-randomness generation \cite{HBK23}, and quantum metrology \cite{CSD19}.

    An RD map for imaginary is given by
    \begin{equation}
        \Delta(\rho)=\frac{1}{2}\big(\rho+\rho^{\mathrm{T}}\big),
    \end{equation}
    where the superscript $(\,\cdot\,)^\mathrm{T}$ denotes transposition with respect to the reference basis.
    Because $\Delta$ is not completely positive it fails to correspond to any physical operation, i.e., it is not an RD channel.
    Notably, even though the set $\mathcal{F}(A)$ is affine, no RD channel for this QRT exists \cite{G17}.
    Remarkably, the conditional censorship protocol establishes an unbreakable censorship of the imaginary.
    This is due to the QRT being affine, and thus Theorem \ref{th:CenAff} ensures that the censorship protocol is unbreakable. 

    One is left with the task of choosing a conditional RD channel $\Delta$ to be deployed in the protocol.
    Let $\Delta$ be the conditional RD channel defined in Eq. \eqref{eq:CRD-Rep}, with 
    \begin{equation}
        \Delta_{m_\sigma}(\rho)=\sum_a\ket{\sigma^a}\bra{\sigma^a}\rho\ket{\sigma^a}\bra{\sigma^a},
    \end{equation}
    as given in Eq. \eqref{eq:deph}.
    As we noted in Sec. \ref{sec:QRT}, $\Delta_{m_\sigma}$ can only function as a conditional RD channel, if the eigenstates $\ket{\sigma^a}$, of any free state $\sigma$, are free as well. 
    This is the case, because $\sigma=\sigma^\mathrm{T}$ and thus $\ket{\sigma^a}\bra{\sigma^a}=\ket{\sigma^a}\bra{\sigma^a}^\mathrm{T}\in\mathcal{F}(A)$ by linearity.
    
    The description of a real-valued quantum state $\sigma$ demands for the specification of $(d+2)(d-1)/2$ real numbers in the state description $m_\sigma$.
    In contrast, $d^2-1$ real numbers are necessary to specify an arbitrary $d\times d$ density operator.
    Note that this is still a non-negligible demand on a users computational expense. 
    To an extend, this was to be expected, because real-valued quantum mechanics can efficiently simulate complex quantum mechanics \cite{KM09}.
    Thus, it should not be possible to find an efficient description $m_\sigma$ of an arbitrary free state $\sigma$. 

    \subsection{Censorship of quantum entanglement}

    Quantum entanglement describes correlations between two or more of the system’s components that cannot be attributed to any classical joint description \cite{W89}. It thus refers to a lack of separability between subsystems.
    A quantum state $\sigma\in\mathcal{D}(A)$, acting on a $K$-partite system $\mathcal{H}_A=(\mathbb{C}^d)^{\otimes K}$, is said to be (fully) separable if it can be written as a probabilistic mixture of (pure) product states \cite{HH09}
    \begin{equation}
        \label{eq:SepState}
        \sigma=\sum_a p_a \sigma^a_1\otimes\dots\otimes \sigma^a_K,
    \end{equation}
    with $p_a\geq 0$ and $\sum_ap_a=1$ being a probability distribution.
    Otherwise, the state is said to be entangled.
    
    In the QRT of entanglement \cite{CL14,CV20}, the set of free states $\mathcal{F}(A)$ contains all separable $K$-partite states of the system $A$. 
    Every convex sum of separable states is again separable; entanglement theory is therefore a convex QRT.
    Mathematically, $\mathcal{F}(A)$ is the convex hull of $\mathcal{D}(\mathbb{C}^d)^{\otimes K}$.
    On the other hand, every $K$-partite quantum state, including entangled ones, can be written as an affine combination of pure product states. 
    From this viewpoint, entanglement theory is maximally nonaffine \cite{CG19}.
    This follows from the fact that pure product states form a basis for the space of linear operators.
    
    Since $\mathcal{F}(A)$ is not affine there is no RD channel for this resource.
    In Ref. \cite{PS24}, a censorship of entanglement was still put forward.
    However, this censorship only preserved a small subset of separable states (those with zero diagonal discord) and was breakable by uncooperative parties in the network. 

    Remarkably, we can establish an unbreakable censorship using the replacement channel $\Delta_{m_\sigma}(\rho)=\mathrm{Tr}(\rho)\sigma$ from Eq. \eqref{eq:repl}.
    The channel $\Delta_{m_\sigma}$ provides us with a conditional RD channel via Eq. \eqref{eq:CRD-Rep}.
    The corresponding censorship protocol is unbreakable due to Theorem \ref{th:CenCon}. 
    The theorem is applicable because $\Delta_{m_\sigma}$ is entanglement breaking.
    Note that, the state description of a separable state $m_\sigma$ contains only the information of $Kd$ amplitudes for each state in an ensemble. 
    Thus, the number of amplitudes to be stored in $m_\sigma$ scales linearly in the number of qudits $K$. 
    Despite being nonlinear, the censorship protocol can be implemented efficiently.

    Note that it is not possible to use the dephasing channel $\Delta_{m_\sigma}=\sum_{a}\ket{\sigma^a}\bra{\sigma^a}(\,\cdot\,)\ket{\sigma^a}\bra{\sigma^a}$ from Eq. \eqref{eq:deph} for the censorship.
    This is due to the eigenstates $\ket{\sigma^a}$, of a separable state $\sigma$, possibly being entangled.
    Thus, $\Delta_{m_\sigma}$ fails to be a conditional RD channel for entanglement.
    
    To see this explicitly, consider the isotropic state
    \begin{equation}
        \sigma=p\phi^++(1-p)\frac{\mathbb{1}^{\otimes 2}}{d^2}
    \end{equation}
    of a two-qudit system, i.e., $K=2$. 
    Here, $\phi^+=\ket{\phi^+}\bra{\phi^+}$ is a maximally entangled state, with $\ket{\phi^+}=\tfrac{1}{\sqrt{d}}\sum_{a=1}^d\ket{aa}$. 
    We consider the state $\sigma$ to be separable, which is the case for $p\leq 1/(1+d)$.
    However, $\sigma$ has still an entangled eigenstate $\ket{\phi^+}$. 
    One immediately verifies that this entangled state is smuggled past the agent when given the description $m_\sigma$ of $\sigma$, viz.
    \begin{equation}
        \begin{split}
            \Delta_{m_\sigma}(\phi^{+})&=\sum_{a}\ket{\sigma^a}\bra{\sigma^a}\phi^{+}\ket{\sigma^a}\bra{\sigma^a},\\
            &=\phi^{+}\notin\mathcal{F}(B),\\
        \end{split}
    \end{equation}
    where we used that one of the eigenstates $\ket{\sigma^a}$ of $\sigma$ corresponds to $\ket{\phi^+}$ and that the other eigenstates are orthogonal to $\ket{\phi^+}$. 
    In this case, the entangled state $\phi^{+}$ reaches the receiver $B$ and the censorship has been broken.

    \subsubsection{Quantum optical realization}

    A physical realization of the censorship can be deployed using polarization optics. 
    For simplicity, we restrict attention to the transmission of pure two-qubit states, i.e., $K=2$ and $d=2$. 
    A sender $A$ prepares the (entangled) Bell state 
    \begin{equation}
        \ket{\psi}=\frac{1}{\sqrt{2}}\big(\ket{HH}+\ket{VV}\big)
    \end{equation}  
    in their lab.
    Here, $\ket{H}$ and $\ket{V}$ denote horizontal and vertical polarization, respectively.
    To prevent the transmission of quantum entanglement to $B$, the agent simply applies two polarization filters to the state $\ket{\psi}$.
    Before doing so, the agent demands a state description $m=(\alpha_{H},\alpha_{V},\beta_{H},\beta_{V})$ corresponding to the free (i.e., separable) state
    \begin{equation*}
        (\alpha_{H}\ket{H}+\alpha_{V}\ket{V})\otimes (\beta_{H}\ket{H}+\beta_{V}\ket{V}).
    \end{equation*}

    Since $A$ cannot give truthful testimony of their Bell state (it is a resource), $A$ pretends to have prepared the separable state $\sigma=\ket{+-}\bra{+-}$, where $\ket{\pm}=(\ket{H}\pm\ket{V})/\sqrt{2}$,
    denotes the diagonal and antidiagonal polarization state, respectively.
    In other words, $A$ gives the description $m_{\sigma}=\big(\tfrac{1}{\sqrt{2}},\tfrac{1}{\sqrt{2}},\tfrac{1}{\sqrt{2}},-\tfrac{1}{\sqrt{2}}\big)$ to the agent. 
    Accordingly, the agent performs their projective measurements $\ket{+}\bra{+}$ and $\ket{-}\bra{-}$ on the first and second qubit, respectively.
    After the measurement, the receiver $B$ is thus left with the state
    $\Delta_{m_{\sigma}}(\ket{\psi}\bra{\psi})=\sigma$
    which is separable.
    The censorship was successful in prohibiting the transmission of the Bell state $\ket{\psi}$.
    In contrast, would $A$ have actually prepared the state $\sigma$ as claimed in their state description $m_\sigma$, then the state would have not been altered, i.e., $\Delta_{m_\sigma}(\sigma)=\sigma$.
    
    \subsection{Censorship of quantum discord}
    
    Quantum discord \cite{OZ01,MB12} is, among others \cite{WP09}, a nonconvex measure of quantum correlations in bipartite systems $\mathcal{H}=\mathcal{X}\otimes\mathcal{Y}$. 
    It is defined as 
    \begin{equation}
        \label{eq:discord}
        \delta(\rho)_{\mathcal{X}\mathcal{Y}}=I(\rho)-\underset{\{P_\mathcal{Y}^a\}_a}{\mathrm{max}} I(\rho^\prime)
    \end{equation}
    where $I(\rho)$ is the quantum mutual information and the maximization is taken over projective measurements $\{P_\mathcal{Y}^a\}_a$, i.e.,
    \begin{equation}
        \rho^\prime=\sum_a\mathrm{Tr}_\mathcal{Y}(\mathbb{1}_\mathcal{X}\otimes P_\mathcal{Y}^a\rho)\otimes \ket{a}\bra{a}_\mathcal{Y}.
    \end{equation}
    Here, $\{\ket{a}\}_a$ denotes an orthonormal basis of $\mathcal{Y}$. 
    The discord $\delta(\rho)_{\mathcal{X}\mathcal{Y}}$ can be interpreted as the correlations that remain when the classical correlations in $\rho$ are subtracted from its total correlations \cite{CG19,OZ01}. 
    Other authors attribute discord to noncommutativity \cite{L08,HFZ12}. 
    Unlike entanglement, discord can exist in separable states as well.
    
    In the QRT of quantum discord, the free states $\sigma\in\mathcal{F}(A)$ are those with vanishing discord, i.e., $\delta(\sigma)_{\mathcal{X}_A\mathcal{Y}_A}=0$.
    It was shown that, a state $\sigma$ has vanishing discord, if and only if, it is a classical-quantum state \cite{D10,BD13}
    \begin{equation}
        \label{eq:CQS}
        \sigma=\sum_{a}q_a \ket{a}\bra{a}_{\mathcal{Y}_A}\otimes \omega^a_{\mathcal{X}_A}.
    \end{equation}
    Here, $\{\ket{a}\}_a$ can be any orthonormal basis of $\mathcal{Y}_{A}$ and $\omega^a\in\mathcal{D}(\mathcal{X}_A)$ are arbitrary quantum states.
    
    The set $\mathcal{F}(A)$ of classical-quantum states is nonconvex.
    It follows that, there is no RD channel for this resource.
    According to Theorem \ref{th:CRD}, there is no conditional RD channel for this theory as well.
    To see explicitly where censorship fails, consider a sender $A$ injecting the state 
    \begin{equation}
        \label{eq:Dis}
        \rho=\sum_{b}p_b\ket{m_{\sigma^b}}\bra{m_{\sigma^b}}\otimes\sigma^{b}_{\mathcal{X}_A\mathcal{Y}_A}
    \end{equation}
    into the circuit of the conditional censorship protocol
    \begin{equation*}
        \Qcircuit @C=1em @R=.7em {
            \lstick{}& \qw & \meter & \cctrlo{1} & \rstick{}\\
            \lstick{\mathcal{X}_A}& \qw & \qw & \multigate{1}{\Delta_m} & \qw & \rstick{\mathcal{X}_B}\qw\\ 
            \lstick{\mathcal{Y}_A}& \qw & \qw & \ghost{\Delta_m} & \qw & \rstick{\mathcal{Y}_B.}\qw
        }
    \end{equation*}
    In Eq. \eqref{eq:Dis}, $\sigma^b$ are considered to be discord-free states of the form \eqref{eq:CQS}, so that they pass the censorship.
    In formula, $\Delta_{m_{\sigma^b}}(\sigma^b)=\sigma^b$, when the correct description $m_{\sigma^b}$ is provided.
    A direct calculation reveals the receiver's state to be
    \begin{equation}
        \label{eq:CD}
        \begin{split}
           \Delta(\rho)&=\sum_{m_{\sigma^a},m_{\sigma^b}}p_a|\braket{m_{\sigma^a}|m_{\sigma^b}}|^2 \Delta_{m_{\sigma^a}}(\sigma^a),\\
           &=\sum_bp_b\Delta_{m_{\sigma^b}}(\sigma^b),\\
           &=\sum_bp_b\sigma^b,
        \end{split}
    \end{equation}
    where we made use of the orthogonality $\braket{m_{\sigma^a}|m_{\sigma^b}}=\delta_{ab}$.  
    Since, discord is a nonconvex resource, the convex sum $\sum_bp_b\sigma^b$ obtained by $B$ may be resourceful.
    Our failure to establish a censorship on discord is due to the fact that probabilistic mixtures constitute a resource in this QRT.

    \subsection{Bell nonlocality}
    \label{ssec:Bell}

    Quantum nonlocality describes our incapability to attribute certain phenomena in quantum physics to any local hidden-variable model \cite{ADA14}, such as the EPR paradox \cite{EPR35}.
    It is known to be an essential ingredient for quantum cryptography \cite{E91}.
    The notion of nonlocality dates back to the seminal work by Bell \cite{B64}, in which it was shown that the assumption of local hidden-variables leads to (Bell-)inequalities which are violated by certain entangled states.

    In the QRT of quantum nonlocality \cite{HH05,LC18}, the free states of a sender $A$ are characterized via the probability distribution they produce from local measurements. 
    More precisely, let $\rho$ be a bipartite state $\rho\in\mathcal{D}(A)=\mathcal{D}(\mathcal{X}\otimes\mathcal{Y})$.
    Consider local POVMs (positive-operator valued measures) $\{M_{a|x}\}_a$ and $\{N_{b|y}\}_b$ for the subsystem $\mathcal{X}$ and $\mathcal{Y}$, respectively. 
    Depending on the choice of POVMs $(x,y)$ and the measurement outcome $(a,b)$, the probability is computed according to Born's rule
    \begin{equation}
        \label{eq:Born}
        p(a,b|x,y)=\mathrm{Tr}\big(\rho(M_{a|x}\otimes N_{b|y})\big).
    \end{equation}
    The state $\rho$ is said to be nonlocal, if the probability \eqref{eq:Born} cannot be written as a probability 
    \begin{equation}
        \label{eq:LHVM}
        p(a,b|x,y)=\sum_{\lambda}q_\lambda p_{\mathcal{X}}(a|x,\lambda)p_{\mathcal{Y}}(b|y,\lambda)
    \end{equation}
    originating from classical local operations and shared randomness. 
    Here, $q_\lambda$ denotes a shared probability distribution, while $p_{\mathcal{X}}(a|x,\lambda)$ and $p_{\mathcal{Y}}(b|y,\lambda)$ are probabilities on the individual subsystem, possibly conditioned on the outcome $\lambda$.

    If a bipartite state $\sigma\in\mathcal{D}(A)$ admits a probability distribution of the form \eqref{eq:LHVM}, then it is said to be local. 
    In the QRT of nonlocality, local states form the set of the free states $\mathcal{F}(A)$.
    It is readily seen that separable states $\sigma=\sum_{\lambda} q_\lambda\omega^\lambda_{\mathcal{X}}\otimes \tau^\lambda_{\mathcal{Y}}$ are local.
    If, we are concerned with pure states only, a pure state $\ket{\psi}$ is nonlocal if and only if it is entangled \cite{GP92}.
    However, there are mixed entangled states that are still local, see, e.g., Refs. \cite{B02,W89}. 
    The resource of nonlocality is therefore distinguished from entanglement.
    In particular, nonlocality can be activated \cite{P12}, that is, there exist local states $\sigma^1,\dots,\sigma^N$ such that $\sigma^1\otimes\dots\otimes\sigma^N$ is nonlocal.
    According to Theorem \ref{th:TPS}, the censorship protocol can be overcome by multiple senders preparing an activated state.

    For concreteness, let each sender prepare the state
    \begin{equation}
        \sigma=p\phi^++(1-p)\frac{\mathbb{1}^{\otimes 2}}{d^2},
    \end{equation}
    where $\phi^+=\tfrac{1}{d}\sum_{a,b=1}^d\ket{aa}\bra{bb}$ is a maximally entangled state and $d$ is the dimension of the subsystems $\mathcal{X}$ and $\mathcal{Y}$.
    It was shown in Refs. \cite{B02,APB07} that $\sigma$ is entangled but local if
    \begin{equation}
        \label{eq:local}
        \frac{1}{1+d} < p\leq \frac{(3d-1)(d-1)^{d-1}}{(d+1)d^{d}}.
    \end{equation}
    For $p\leq 1/(1+d)$ the state $\sigma$ is separable.
    Since, $\sigma$ is a free (i.e., local) state we have $\Delta(\ket{m_\sigma}\bra{m_\sigma}\otimes \sigma)=\sigma$, due to relation (iv).
    If each of the senders $A_1,\dots,A_N$ individually prepares the free state $\sigma$ together with the correct description $m_\sigma$,
    the state obtained by the receivers $B_1,\dots,B_N$ is simply a product state
    \begin{equation}
        \label{eq:nonloc}
        \Delta\big(\ket{m_{\sigma}}\bra{m_{\sigma}}\otimes\sigma\big)^{\otimes N}\\
            =\left(p\phi^++(1-p)\frac{\mathbb{1}^{\otimes 2}}{d^2}\right)^{\otimes N}.
    \end{equation}
    For $p$ satisfying the upper bound in Eq. \eqref{eq:local} and a large number of users $N$, the state \eqref{eq:nonloc} is nonlocal; see Ref. \cite{P12} for a proof.
    The censorship has been broken.

    It is of course possible to establish a stricter censorship that is unbreakable, e.g., a censorship of entanglement or a censorship of coherence. 
    This, however, comes at the cost of certain local states to be prohibited from transmission as well.

    \section{Conclusion}
    \label{sec:FIN}

    In this paper, we introduced a protocol for quantum censorship. 
    Therein, an agent applies conditional RD channels locally to each sender-receiver pair. 
    This prohibits the distribution of resource states in a public-domain network. 
    By using conditional channels, the protocol can censor a much wider variety of resources than any RD channel. 
    We saw this explicitly for the QRT of imaginarity and entanglement.
    For both resources, no RD channel exists, but conditional RD channels establish an unbreakable censorship.
    On the other hand, we found that the protocol fails to censor discord and nonlocality.
    The advantage conditional RD channels offer, comes at the cost of senders providing a classical description of their state to be transmitted.
    In general, such a cost amounts to an exponential overhead in classical resources, which is the reason why in general nonlinear quantum operations are unfeasible in many settings.
    However, in the censorship protocol users are only permitted to transmit free states.
    Thus, their state description can often be given without imposing high computational expense on the side of the users.
    Moreover, our protocol avoids tomography procedures which are unreliable when senders are uncooperative.
    In addition, the protocols do not rely on multiple copies of quantum states to perform nonlinear quantum operations, which would make censorship extremely costly. 
    
    Developing quantum censorship protocols becomes especially relevant once we are confronted with the emergence of a widely accessible quantum internet. 
    On the one hand, quantum censorship allows the prevention of malicious parties from quantum-cryptographic attacks in places where post-quantum cryptography is not at the state-of-the-art. 
    On the other hand, recent experimental progress forecasts a possible usage of the preexisting telecommunication networks for the transmission of quantum information \cite{ZCB21,LBF22,JX23}.
    In such a setting, commercial enterprises may offer free classical services but want to charge premium fees for quantum communication.
    The censorship protocol would enable such policies.
    For a quantum internet, in which quantum resources are ubiquitous, the control and regulation of quantum communication might become a primary concern in quantum information science.
    We hope this work paves the way for quantum censorship as an emerging topic in quantum information theory.

    \acknowledgments
    
    We gratefully acknowledge financial support from Denmarks Grundforskningsfond (DNRF 139, Hy-Q Center for Hybrid Quantum Networks) and the Alexander von Humboldt Foundation (Feodor Lynen Research Fellowship).

    \appendix 

    \section{Representation of conditional RD channels}
    \label{app:rep}
    Here, we prove that the representation in Eq. \eqref{eq:CRD-Rep} is both necessary and sufficient for having a conditional RD channel.
    \begin{lemma}
        \label{lem:CRD}
        A quantum channel $\Delta:\mathcal{D}(MA)\to\mathcal{D}(B)$ is a conditional RD channel, if and only if, it can be written as 
        \begin{equation}
        \label{eq:app-1}
        \Delta=\sum_{m_\sigma}\braket{m_\sigma|(\,\cdot\,)|m_\sigma}\otimes \Delta_{m_{\sigma}},
    \end{equation}
    where $\{\ket{m_{\sigma}}\}_{m_{\sigma}}$ is an orthonormal basis and $\Delta_{m_{\sigma}}$ obeys 
    \begin{align*}
        \text{(v)}\quad\forall\rho\in\mathcal D(A):& \quad\Delta_{m_\sigma}(\rho)\in\mathcal{F}(B),
        \\
        \text{(vi)}\quad\forall\sigma\in\mathcal F(A):& \quad\Delta_{m_{\sigma}}(\sigma)=\sigma,
    \end{align*}
    where $\mathcal{F}$ is the (convex) set of free states.
    \end{lemma}
    \begin{proof}
        First, we show sufficiency, i.e., Eq. \eqref{eq:app-1} satisfies Definition \ref{def:CRD}.
        Expanding a state $\rho\in\mathcal{D}(MA)$ as in Eq. \eqref{eq:exp} yields 
        \begin{equation}
        \label{eq:app-2}
            \begin{split}
                \Delta(\rho)&=\sum_{a,b,c}t_{bc}\braket{m_{\sigma^a}|m_{\sigma^b}}\braket{m_{\sigma^c}|m_{\sigma^a}}\Delta_{m_{\sigma^a}}(\rho^{bc}),\\
                &=\sum_{a} t_{aa} \Delta_{m_{\sigma^a}}(\rho^{aa}),\\
            \end{split}
        \end{equation}
        where we used orthogonality, i.e., $\braket{m_{\sigma^a}|m_{\sigma^b}}=\delta_{ab}$.
        It follows from (v), together with $t_{aa}\geq 0$, and $\mathcal{F}(B)$ being convex, that
        \begin{equation}
            \Delta(\rho)= \sum_{a} t_{aa}\Delta_{m_{\sigma^a}}(\rho^{aa})\in\mathcal{F}(B)
        \end{equation}
        is a free state. 
        Thus, property (iii) of a conditional RD channel is verified.
        Similarly, condition (vi) yields
        \begin{equation}
            \begin{split}
                \Delta(\ket{m_\sigma}\bra{m_\sigma}\otimes\sigma)&=\Delta_{m_\sigma}(\sigma),\\
                &=\sigma,
            \end{split}
        \end{equation}
        for any free state $\sigma\in\mathcal{F}(A)$.
        Thus, property (iv) is verified and $\Delta$ is shown to be a conditional RD channel.

        To prove necessity, suppose $\Delta$ is a conditional RD channel; see Definition \ref{def:CRD}.
        Any channel $\Delta:\mathcal{D}(MA)\to\mathcal{D}(B)$ can be written as 
        $\Delta=\sum_a \Gamma_a\otimes \Delta_a$, where $\Gamma_a:\mathcal{D}(M)\to\mathbb{C}$ and $\Delta_a:\mathcal{D}(A)\to\mathcal{D}(B)$ are completely positive maps.
        In particular, $\Gamma_a(\rho)=\sum_b\braket{\psi_{ab}|\rho|\psi_{ab}}$ for some (not necessarily normalized) $\ket{\psi_{ab}}\in\mathcal{H}_M$.
        Because $\Delta$ is, by assumption, a conditional RD channel, property (iv) holds.
        This implies 
        \begin{equation*}
            \begin{split}
                \forall\sigma\in\mathcal{F}(A):~ \Delta(\ket{m_\sigma}\bra{m_\sigma}\otimes\sigma)&=\sum_{a,b}|\braket{\psi_{ab}|m_{\sigma}}|^2\Delta_a(\sigma),\\
                &=\sigma.\\
            \end{split}
        \end{equation*}
        Because the above equality holds for an orthonormal basis $\{\ket{m_\sigma}\}_{m_\sigma}$, it follows that $\forall b:\ket{\psi_{ab}}\propto\ket{m_{\sigma^a}}$.  
        Then, we must have that $\Delta_{m_\sigma}(\sigma)=\sigma$ for all $\sigma\in\mathcal{F}(A)$, which implies condition (vi).
        It remains to show that condition (v) holds as well.
        To do so, let $\rho\in\mathcal{D}(MA)$ be an arbitrary quantum state.
        Using the expansion \eqref{eq:exp}, it follows that $\Delta(\rho)=\sum_{a} t_{aa} \Delta_{m_{\sigma^a}}(\rho^{aa})$.
        Since, $\Delta$ satisfies property (iii), i.e., $\Delta(\rho)\in\mathcal{F}(B)$ for any $\rho\in\mathcal{D}(MA)$, we must have that $\Delta_{m_{\sigma^a}}(\rho^{aa})\in\mathcal{F}(B)$ is a free state as well, for any $\rho^{aa}\in\mathcal{D}(A)$.
        Thus, condition (v) is necessary as well, concluding the proof.
    \end{proof}
    
    \section{Purity is non-increasing under unital channels}
    \label{app:pur}
    In this part of the appendix we show that the purity of a quantum state is non-increasing under unital channels.
    \begin{lemma}
        \label{lem:pur}
        Let $\rho\in\mathcal{D}(\mathcal{H})$ and let $\Lambda:\mathcal{D}(\mathcal{H})\to \mathcal{D}(\mathcal{H})$ be a unital channel.
        It holds that $\mathrm{Tr}[\Lambda(\rho)^2]\leq \mathrm{Tr}(\rho^2)$.
    \end{lemma}
    For concreteness, let $\mathcal{H}$ be a $d$-dimensional Hilbert space.
    A quantum channel $\Lambda$ is unital if $\Lambda(\mathbb{1}/d)=\mathbb{1}/d$.
    For proving Lemma \ref{lem:pur}, it is useful to review the notion of majorization \cite{W18}. 
    A quantum state $\rho$ majorizes a quantum state $\sigma$ if
    \begin{equation}
        \sum_{a=1}^k \lambda_a(\sigma)\leq\sum_{a=1}^k\lambda_a(\rho),
    \end{equation}
    for all $k=1,\dots,d$.
    Here, $\lambda_a(\rho)$ denotes the $a$th ordered eigenvalue of $\rho$, i.e., $\lambda_1\geq \dots\geq \lambda_d$.
    If $\rho$ majorizes $\sigma$, we write $\sigma\preceq \rho$.
    We now proceed with the proof of Lemma \ref{lem:pur}.
    \begin{proof}
        First, note that the purity of $\rho$ is given by $\mathrm{Tr}(\rho^2)=\sum_{a=1}^d\lambda_a(\rho)^2$. 
        Furthermore $\sigma\preceq\rho$ implies $\sum_{a=1}^d \lambda_a(\sigma)^2\leq\sum_{a=1}^d\lambda_a(\rho)^2$.
        It follows that $\mathrm{Tr}(\sigma^2)\leq \mathrm{Tr}(\rho^2)$.
        Moreover, it holds that $\sigma\preceq \rho$, if and only if, $\sigma=\Lambda(\rho)$ for some unital channel $\Lambda$ \cite{W18}, a result originally due to Uhlmann \cite{U72}.
        This implies that for any unital channel $\Lambda$ we have $\Lambda(\rho)\preceq \rho$.
        Thus, we arrive at the desired result $\mathrm{Tr}[\Lambda(\rho)^2]\leq \mathrm{Tr}(\rho^2)$.
    \end{proof}

    Finally, we note that the dephasing channel 
    \begin{equation}
        \Delta(\rho)=\sum_{a=1}^d \ket{\sigma^a}\bra{\sigma^a}\rho\ket{\sigma^a}\bra{\sigma^a},
    \end{equation}
    with an orthonormal basis $\{\ket{\sigma^a}\}_{a}$, is unital. 
    That is $\Delta(\mathbb{1}/d)=\tfrac{1}{d}\sum_a\ket{\sigma^a}\bra{\sigma^a}=\mathbb{1}/d$.
    Applying Lemma \ref{lem:pur}, we arrive at $\mathrm{Tr}[\Delta(\rho)^2]\leq \mathrm{Tr}(\rho^2)$, which justifies Eq. \eqref{eq:non-inc} in Sec. \ref{sec:GQC}.

\end{document}